\documentclass{article}

\usepackage{graphicx}
\usepackage{amssymb}
\usepackage{amsthm}
\usepackage{amsmath}
\usepackage{amsfonts}
\usepackage[T1]{fontenc}

\title{Swapping Colored Tokens on Graphs}
\author{
  Katsuhisa Yamanaka\thanks{Iwate University, Japan.} \and
  Takashi Horiyama\thanks{Saitama University, Japan.} \and
  J. Mark Keil\thanks{University of Saskatchewan, Canada.} \and
  David Kirkpatrick\thanks{University of British Columbia, Canada.} \and 
  Yota Otachi\thanks{Kumamoto University, Japan.} \and
  Toshiki Saitoh\thanks{Kyshu Institute of Technology, Japan.} \and
  Ryuhei Uehara\thanks{Japan Advanced Institute of Science and Technology, Japan.} \and
  Yushi Uno\thanks{Osaka Prefecture University, Japan.}
}
\date{}

\newtheorem{theorem}{Theorem}
\newtheorem{lemma}[theorem]{Lemma}
\newtheorem{corollary}[theorem]{Corollary}

\newcommand{\textem}[1] {{\em #1}}
\newcommand\mset[1]{\{#1\}}
\newcommand\msize[1]{\left|#1\right|}
\newcommand\seq[1]{\langle#1\rangle}
\newcommand\order[1]{$\mbox{\cal O}(#1)$}

\newcommand\cts[1]{\textsc{$#1$-Colored Token Swapping}}
\newcommand\ncts{\textsc{Colored Token Swapping}}
\newcommand\tsw{\textsc{Token Swapping}}
\newcommand\ptdm{\textsc{Planar 3DM}}
\newcommand\tdm{\textsc{3DM}}
\newcommand\poilp{\textsc{$p$-Opt-ILP}}

\newcommand{\colorset}{C}
\newcommand{\ncolors}{c}
\newcommand{\cycle}{A}
\newcommand{\tpl}{token-placement}
\newcommand{\tpls}{token-placements}
\newcommand{\diff}[1]{\textrm{diff}(#1)}

\newcommand{\ntokensinitial}{n^1(\initialf)}
\newcommand{\ntokenstarget}{n^1(\targetf)}
\newcommand{\givenint}{\ell}

\newenvironment{mycase}[2]{%
\vspace{3mm} \noindent \textbf{Case~#1}: #2\par}{}

\newcommand{\mapf}{f}
\newcommand{\mapg}{g}
\newcommand{\mapfp}{\mapf^\prime}
\newcommand{\initialf}{\mapf_0}
\newcommand{\targetf}{\mapf_t}

\newcommand{\calS}{{\mathcal S}}

\newcommand{\lenS}{\textrm{len}}
\newcommand{\OPT}[2]{\textrm{OPT}(#1, #2)}

\newcommand{\dgraph}[2]{D(#1,#2)}
\newcommand{\dvset}{V_{D}}
\newcommand{\deset}{E_{D}}

\newcommand{\ccover}[2]{{\cal C}(#1,#2)}
\newcommand{\ccoverp}[2]{{\cal C}'(#1,#2)}
\newcommand{\ccoveropt}[2]{{\cal C}_{OPT}(#1,#2)}

\newcommand{\pfunc}[2]{\Phi(#1,#2)}

\newenvironment{listing}[1]{%
        \begin{list}{*}{%
                 \settowidth{\labelwidth}{#1}%
                 \setlength{\leftmargin}{\labelwidth}%
                  \advance \leftmargin by 12pt
                   \setlength{\itemsep}{0pt}%
                   \setlength{\parsep}{0pt}%
                   \setlength{\topsep}{0pt}%
                   \setlength{\parskip}{0pt}%
}%
}{%
\end{list}}

\begin{document}
\maketitle
\begin{abstract}
We investigate the computational complexity of the following problem.
We are given a graph in which each vertex has an initial and a target color.
Each pair of adjacent vertices can swap their current colors.
Our goal is to perform 
the minimum number of swaps
so that the current and target colors agree at each vertex.
When the colors are chosen from $\mset{1,2,\dots,c}$,
we call this problem \cts{\ncolors} since the current color of a vertex can be seen as a colored token placed on the vertex.
We show that \cts{\ncolors} is NP-complete for $\ncolors = 3$
even if input graphs are restricted to
connected planar bipartite graphs of maximum degree 3.
We then show that \cts{2} can be solved in polynomial time for general graphs
and in linear time for trees.
Besides, we show that,
the problem for complete graphs is fixed-parameter tractable
when parameterized by the number of colors,
while it is known to be NP-complete when the number of colors is unbounded.
\end{abstract}

\section{Introduction}

Sorting problems are fundamental and important in computer science.
Let us consider the problem of sorting a given permutation by applying the minimum number of swaps of two elements.
If we are allowed to swap only adjacent
elements (that is, we can apply only adjacent transpositions),
then the minimum number of swaps is equal to
the number of inversions of
a permutation~\cite{J85,K98}.
If we are allowed to swap any two elements,
then the minimum number of swaps is equal to
the number of elements of a permutation minus
the number of cycles of
the permutation~\cite{C1849,J85}.
Now, if we are given a set of ``allowed swaps'',
can we exactly estimate the minimum number of swaps?
We formalize this question as
a problem of swapping tokens on graphs.

Let $G=(V,E)$ be an undirected unweighted graph
with vertex set $V$ and edge set $E$.
Suppose that each vertex in $G$ has a token,
and each token has a color in $\colorset=\mset{1,2, \ldots, \ncolors}$.
Given two token-placements,
we wish to transform one to the other by applying the fewest number of token swaps on adjacent vertices.
We call the problem \cts{\ncolors} if $c$ is a constant,
otherwise,
we simply call it \ncts
(a formal definition can be found in the next section).
See \figurename~\ref{fig:example} for an example.
\begin{figure}[tb]
\centerline{\includegraphics[width=0.8\linewidth]{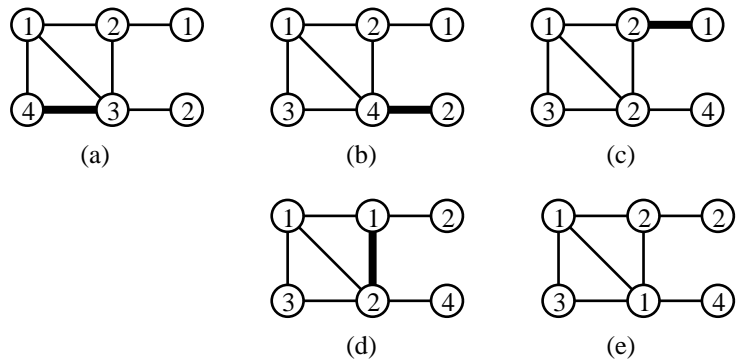}}
\caption{%
An instance of \cts{4}. Tokens on vertices are written inside circles. We swap the two tokens along each thick edge. (a) The initial token-placement. (b)--(d) Intermediate \tpls. (e) The target token-placement.
}
\label{fig:example}
\end{figure}

In this paper, we study the computational complexity of \cts{\ncolors}.
We consider the case where $\ncolors$ is a fixed constant and show the following results.
If $c=2$, then the problem can be solved in polynomial time
(Theorem~\ref{thm:generalgraphs}).
On the other hand, \cts{3} is NP-complete even
for planar bipartite graphs of maximum degree 3
(Theorem~\ref{thm:cts3}).
We also show that \cts{c} is \order{n^{c+2}}-time
solvable for graphs of maximum degree at most 2 (Theorem~\ref{thm:xp}),
\cts{2} is linear-time solvable for trees (Theorem~\ref{thm:tree}), and
\cts{c} is fixed-parameter tractable for complete graphs
if $c$ is the parameter (Theorem~\ref{thm:clique}).

If the tokens have distinct colors, then the problem is called \tsw~\cite{YEIKKOSSUT15}.
This variant has been investigated for several graph classes.
\tsw\ can be solved in polynomial time for paths~\cite{J85,K98},
cycles~\cite{J85}, stars~\cite{I99}, complete graphs~\cite{C1849,J85}, and complete bipartite graphs~\cite{YEIKKOSSUT15}.
Heath and Vergara~\cite{HV03} gave
a polynomial-time 2-approximation algorithm for squares of paths (see also \cite{FMS04,FSL06}).
Yamanaka~et~al.~\cite{YEIKKOSSUT15} gave 
a polynomial-time 2-approximation algorithm for trees.

Recently, Miltzow~et al.~\cite{MNORTU16esa} gave computational complexity results for \tsw.
First, they showed the NP-completeness of \tsw.
Second, they proposed an exact exponential algorithm for
\tsw\ on general graphs
and showed that \tsw\ cannot be solved in $2^{o(n)}$ time
unless the Exponential Time Hypothesis (ETH)
fails, where $n$ is the number of vertices.
Third, they proposed polynomial-time 4-approximation algorithms for \tsw\ on general graphs
and showed the APX-hardness of \tsw.
Independently, Kawahara~et al.~\cite{KawaharaSY17}
proved the NP-completeness of \tsw\ even if an input graph
is bipartite and has only vertices of degree at most 3.
More recently, Bonnet~et al.~\cite{BonnetMR17} gave some
results for \tsw.
They showed the parameterized hardness: \tsw\ is W[1]-hard,
parameterized by the number of swaps and \tsw\ cannot be
solved in $f(k)(n+m)^{o(k/\log{k})}$ unless
the ETH fails,
where $n$ is the number of vertices,
$m$ is the number of edges and
$k$ is the number of swaps.
\tsw\ on trees is one of
the attractive open problems.
They gave an interesting result for the open problem:
\tsw\ is NP-hard even when both the treewidth and
the diameter are constant, and cannot be solved in
$2^{o(n)}$ time unless the ETH fails.

Miltzow~et al.~\cite{MNORTU16esa} and
Bonnet~et al.~\cite{BonnetMR17} also gave important results on
\ncts.
Here, we mention only the results related to our contributions.
Miltzow~et al.~\cite{MNORTU16esa} gave the NP-completeness of
\ncts\ using $\Omega(n)$ colors.
On the other hand, in this paper, we show that
the NP-completeness of \cts{c} even if the number of colors
is only 3.
Bonnet~et al.~\cite{BonnetMR17} showed the NP-completeness of
\ncts\ on complete graphs using $\Omega(n)$ colors.
To complement their result,
we show the fixed-parameter tractability of \cts{c} on complete graphs,
parameterized by the number of colors.

\section{Preliminaries}
\label{subsec:preliminaries}

The graphs considered in this paper are finite, simple, and undirected.
Let $G=(V,E)$ be an undirected unweighted graph with vertex set $V$ and edge set $E$.
We sometimes denote by $V(G)$ and $E(G)$ the vertex set and the edge set of $G$, respectively. 
We always denote $|V|$ by $n$.
For a vertex $v$ in $G$, let $N(v)$ be the set of all neighbors of $v$.

We formalize our problem as a problem reconfiguring an initial
coloring of vertices to the target one by repeatedly swapping
the two colors on adjacent vertices as follows.
Let $\colorset=\mset{1,2, \ldots, \ncolors}$ be a set of colors.
  In this paper, we assume that $c$ is a constant unless otherwise noted.
A \textem{\tpl} of $G$ is a surjective function $\mapf\colon V \to \colorset$.
For a vertex $v$, $\mapf(v)$ represents the color of the token placed on $v$.
Note that we assume that each color in $\colorset$ appears at least once.
Two distinct \tpls\ $\mapf$ and $\mapfp$ of $G$ are \textem{adjacent} if the following two conditions (a) and (b) hold:
\begin{listing}{aaa}
\item[(a)] there exists $(u,v) \in E$ such that $\mapfp(u) = \mapf(v)$ and $\mapfp(v) = \mapf(u)$;
\item[(b)] $\mapfp(w) = \mapf(w)$ for all vertices $w \in V \setminus \{u,v\}$.
\end{listing}
In other words, the \tpl\ $\mapfp$ is obtained from $\mapf$ by {\em swapping} the tokens on the two adjacent vertices $u$ and $v$.
For two \tpls\ $\mapf$ and $\mapfp$ of $G$, a sequence
$\calS = \seq{\mapf_1, \mapf_2, \ldots , \mapf_{h}}$
of \tpls\ is a \textem{swapping sequence}
between $\mapf$ and $\mapfp$
if the following three conditions (1)--(3) hold:
\begin{listing}{aaa}
\item[(1)] $\mapf_1 = \mapf$ and $\mapf_h=\mapfp$;
\item[(2)] $\mapf_k$ is a \tpl\ of $G$ for each $k = 1,2, \ldots ,h$;
\item[(3)] $f_{k-1}$ and $f_{k}$ are adjacent for every $k = 2, 3, \ldots, h$.
\end{listing}
The \textem{length} of a swapping sequence $\calS$,
denoted by $\lenS(\calS)$,
is defined to be the number of \tpls\ in $\calS$ minus
one, that is, $\lenS(\calS)$ indicates the number of swaps in $\calS$.
For two \tpls\ $\mapf$ and $\mapfp$ of $G$, we denote by $\OPT{\mapf}{\mapfp}$ the minimum length of a swapping sequence between $\mapf$ and $\mapfp$.
If there is no swapping sequence between $f$ and $f'$, then we set $\OPT{f}{f'} = \infty$.

Given two \tpls\ $\initialf$ and $\targetf$ of a graph $G$ and a nonnegative integer $\givenint$, 
the \cts{\ncolors} problem is to determine whether $\OPT{\initialf}{\targetf} \leq \givenint$ holds.
From now on, we always denote by $\initialf$ and $\targetf$ the \textem{initial} and \textem{target} \tpls\ of $G$, respectively.

Let $f$ and $f'$ be two token-placements of $G$.
If there exist a component $C$ of $G$ and a color $i$ such that
the numbers of tokens of color $i$ in $C$ are not equal in $f$ and $f'$,
then we cannot transform $f$ into $f'$.
If there are no such $C$ and $i$, then we write $f \simeq f'$.
Note that one can easily check whether $f \simeq f'$ in linear time.
Thus we can assume that input instances satisfy this condition.
It holds that if $f \simeq f'$, then $f$ can be transformed into $f'$ with at most $\binom{n}{2}$ swaps.
This can be shown by slightly modifying the proof of Theorem 1 in~\cite{YEIKKOSSUT15}.

\begin{lemma}\label{lem:polylength}
  Let $\initialf$ and $\targetf$ be token-placements of $G$.
  If $\initialf \simeq \targetf$, then
  \[
  \OPT{\initialf}{\targetf} \le \binom{n}{2}.
  \]
\end{lemma}

\begin{proof}
  No two tokens are swapped twice or more in a swapping
  sequence with the minimum length.
  Hence, the claim holds.
\end{proof}

The bound in Lemma~\ref{lem:polylength} is tight.
For the path graph $(v_1,v_2, \ldots ,v_n)$, we set $\initialf(v_i) = i$ and $\targetf(v_i) = n-i+1$ for $i=1,2, \ldots ,n$.
It is known that this instance requires ${n \choose 2}$ swaps~\cite{J85,K98}.

\section{Hardness results}

In this section, we show that \cts{3} is NP-complete
by constructing a polynomial-time reduction from \ptdm~\cite{DF86}.
To define \ptdm, we first introduce the following well-known 
NP-complete problem.

\medskip
\noindent
\textbf{Problem:} \textsc{3-Dimensional Matching} (\tdm) \cite[SP1]{GJ79}\\
\textbf{Instance:} 
Set $T \subseteq X \times Y \times Z$, where $X$, $Y$, and $Z$
are disjoint sets having the same number $m$ of elements.\\
\textbf{Question:} 
Does $T$ contain a matching, i.e., a subset $T' \subseteq T$ such that $\msize{T'}=m$ and it contains all elements of $X$, $Y$, and $Z$?
\medskip

\ptdm\ is a restricted version of \tdm\ in which
the following bipartite graph $G$ is planar.
The graph $G$ has the vertex set $V(G) = T \cup X \cup Y \cup Z$
with a bipartition $(T, X \cup Y \cup Z)$.
Two vertices $t \in T$ and $w \in X \cup Y \cup Z$ are adjacent in $G$ if and only if $w \in t$.
\ptdm\ is NP-complete even if $G$ is a connected graph of maximum degree 3~\cite{DF86}.

\begin{theorem}
\label{thm:cts3}
\cts{3} is NP-complete even for connected planar bipartite graphs of maximum degree 3.
\end{theorem}
\begin{proof}
By Lemma~\ref{lem:polylength},
there is a polynomial-length swapping sequence between two \tpls,
and thus \cts{3} is in NP.

Now we present a reduction from \ptdm,
as illustrated in \figurename~\ref{fig:reduction}.
Let $(X, Y, Z; T)$ be an instance of \ptdm\ and $m=\msize{X}=\msize{Y} = \msize{Z}$.
As mentioned above, we construct a bipartite graph
$G = (T, X \cup Y \cup Z; E)$ from $(X, Y, Z; T)$.
We
set $\initialf(x) = 2$ and $\targetf(x) = 1$ for every $x \in X$,
set $\initialf(y) = 3$ and $\targetf(y) = 2$ for every $y \in Y$,
set $\initialf(z) = 1$ and $\targetf(z) = 3$ for every $z \in Z$, and
set $\initialf(t) = 1$ and $\targetf(t) = 1$ for every $t \in T$.
From the assumptions, $G$ is a planar bipartite graph of maximum degree 3.
The reduction can be done in polynomial time.
\begin{figure}[tb]
 \centering
 \includegraphics[scale=.8]{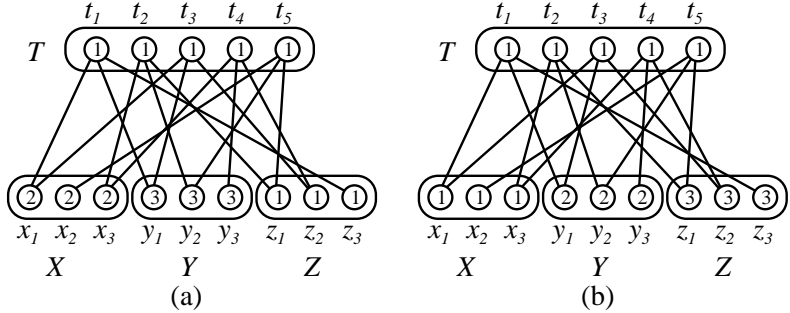}
 \caption{%
(a) The initial \tpl\ and (b) the target \tpl\ of the graph constructed from an instance
 $(X=\mset{x_1,x_2,x_3},\; Y=\mset{y_1,y_2,y_3},\;  Z=\mset{z_1,z_2,z_3},\;  T=\mset{t_1=(x_1,y_1,z_3),\;  t_2=(x_3,y_2,z_1),\; t_3=(x_1,y_1,z_2),\; t_4=(x_3,y_3,z_2),\; t_5=(x_2,y_2,z_1)})$.}
\label{fig:reduction}
\end{figure}
We prove that
the instance $(X,Y,Z; T)$ is a yes-instance if and only if
$\OPT{\initialf}{\targetf} \leq 3m$.

To show the only-if part, assume that there exists a subset $T'$ of $T$
such that $\msize{T'}=m $ and $T'$ contains all elements of $X$, $Y$, and $Z$.
Since the elements of $T'$ are pairwise disjoint,
we can cover the subgraph of $G$ induced by $T' \cup X \cup Y \cup Z$
with $m$ disjoint stars of four vertices,
where each star is induced by an element $t$ of $T'$ and its three elements.
To locally move the tokens to the target places in such a star,
we need only three swaps.
See \figurename~\ref{fig:star-swap}.
This implies that a swapping sequence of length $3m$ exists.
\begin{figure}[tb]
  \centering
  \includegraphics[scale=.7]{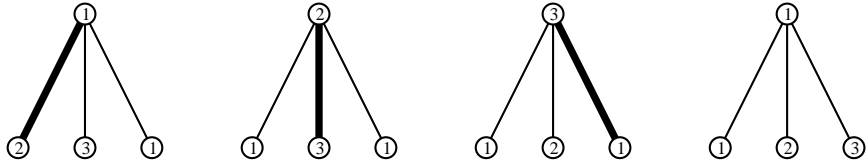}
  \caption{A swapping sequence to resolve the \tpl\ of a triple.}
  \label{fig:star-swap}
\end{figure}

To show the if part, assume that
there is a swapping sequence $\calS$ from $\initialf$ to $\targetf$ 
with at most $3m$ swaps.
Let $T' \subseteq T$ be the set of vertices such that the tokens on them are moved in $\calS$.
Let $G'$ be the subgraph of $G$ induced by $T' \cup X \cup Y \cup Z$.
Let $w \in X \cup Y \cup Z$.
Since $\initialf(w) \ne \targetf(w)$ and $N(w) \subseteq T$,
the sequence $\calS$ swaps the tokens on $w$ and on a neighbor $t \in T'$ of $w$ at least once.
This implies that $w$ has degree at least 1 in $G'$.
Since each $t \in T'$ has degree at most 3 in $G'$,
we can conclude that $\msize{T'} \ge \frac{1}{3} \msize{X \cup Y \cup Z} =  m$.
In $\calS$, the token placed on a vertex in $X \cup Y$ in the initial \tpl\ is moved at least twice,
while the token placed on a vertex in $Z \cup T'$ is moved at least once.
As a swap moves two tokens at the same time,
\[
 \lenS(\calS) \ge \frac{1}{2}(2\msize{X} + 2\msize{Y} + \msize{Z} + \msize{T'}) \ge 3m.
\]
From the assumption that $\lenS(\calS) \leq 3m$, it follows that $\msize{T'} = m$,
and hence each $w \in X \cup Y \cup Z$ has degree exactly 1 in $G'$.
Therefore, $G'$ consists of $m$ disjoint stars centered at the vertices of $T'$
which form a solution of \ptdm.
\end{proof}

The proof above can be modified to show hardness for instances with more colors.
We transform the reduced instance by adding a path of polynomial length containing
additional colors as follows.

\begin{corollary}
\label{cor:ctsk}
  For every constant $c \ge 3$,
  \cts{c} is NP-complete even for connected planar bipartite graphs of maximum
  degree~3.
\end{corollary}
\begin{proof}
  Let $G$ be the graph obtained in the proof of Theorem~\ref{thm:cts3}.
  Recall that the token-placement of $G$ uses three colors 1, 2, and 3.
  It is known that we can assume that $G$ has a degree-2 vertex~\cite{DF86}.
  We connect a path $(p_{4}, p_{5}, \dots, p_c)$ to $G$
  by adding an edge between $p_4$ and a degree-2 vertex in $G$.
  We set $\initialf(p_{i}) = \targetf(p_{i}) = i$ for
  every $i \in \{4, \dots, c\}$.
  The proof of Theorem~\ref{thm:cts3} still works for the obtained graph.
\end{proof}

\section{Positive results}

In this section, we give some positive results.
First, we show that \cts{c} for graphs of maximum degree at most 2
is in XP\footnote{
  See \cite[p.~13]{CyganFKLMPPS15} for a formal definition of XP.
} when $c$ is the parameter. 
Second, we show that \cts{2} for general graphs can be
solved in polynomial time.
Third, we show that the \cts{2} problem for trees
can be solved in linear time
without constructing a swapping sequence.
Finally, we show that the \cts{c} problem for complete graphs
is fixed-parameter tractable\footnote{
  Also see \cite[p.~13]{CyganFKLMPPS15} for a formal definition of
  fixed-parameter tractability.
}
when $c$ is the parameter.

\subsection{\cts{c} on graphs of maximum degree 2}

In this subsection, we show that the degree bound in Theorem~\ref{thm:cts3} and Corollary~\ref{cor:ctsk} is tight.
If a graph has maximum degree at most 2, then we can solve \cts{c} in XP time for every constant $c$ as follows.
Each component of a graph of maximum degree at most 2 is a path or a cycle.
Observe that a shortest swapping sequence does not swap tokens of the same color.
This immediately gives a unique matching between tokens and target vertices for a path component.
For a cycle component, observe that each color class has at most $n$ candidates for such a matching restricted to the color class.
This is because after we guess the target of a token in a color class, the targets of the other tokens in the color class can be uniquely determined.
In total, there are at most $n^{c}$ matchings between tokens and target vertices.
By guessing such a matching, we can reduce \cts{c} to \tsw.
Now we can apply Jerrum's \order{n^{2}}-time algorithms
for solving \tsw\ on paths and cycles~\cite{J85}.
Thus we have the following theorem.
\begin{theorem}\label{thm:xp}
  \cts{c} can be solved in \order{n^{c + 2}} time for graphs of maximum degree at most 2.
\end{theorem}

\subsection{\cts{2} on general graphs}

In this section, we show that \cts{2} is polynomially solvable.
We lately noticed that an equivalent problem was solved previously by van den Heuvel~\cite[Theorem 4.3]{Heuvel13}.
Our proof is quite similar to the one in \cite{Heuvel13}.
However, to be self-contained we describe our proof.

Let $\colorset = \mset{1,2}$ be the color set. 
Let $G=(V,E)$ be a graph, and let $\initialf$ and $\targetf$
be initial and target \tpls.
We first construct the following weighted 
complete bipartite graph $G_B = (X,Y,E_B,w)$ as follows. 
The vertex sets $X, Y$ and the edge set $E_B$ are defined as follows:
\begin{eqnarray*}
X &=& \mset{x_v \mid v \in V \mbox{ and }\initialf(v) = 1}, \\
Y &=& \mset{y_v \mid v \in V \mbox{ and }\targetf(v) = 1}, \\
E_B & = & \mset{(x,y) \mid x\in X \mbox{ and } y \in Y}.
\end{eqnarray*}
The vertices in $X$ correspond to the vertices in $V$
having tokens of color 1 in $\initialf$,
and the vertices in $Y$ correspond to the vertices in $V$
having tokens of color 1 in $\targetf$.
For $x \in X$ and $y \in Y$, the weight $w(e)$ of the edge $e = (x, y)$ is
defined as the length of a shortest path from $x$ to $y$ in $G$.
\begin{figure}[tb]
  \centerline{\includegraphics[width=0.9\linewidth]{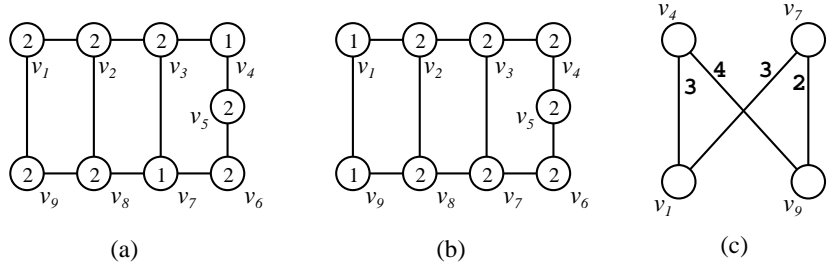}}
  \caption{(a) The initial \tpl. (b) The target \tpl. (c) The weighted complete bipartite graph constructed from (a) and (b). The edge weights are written beside the edges.}
\label{fig:matching}
\end{figure}
\figurename~\ref{fig:matching} 
gives an example of an initial \tpl, a target \tpl, and the associated weighted complete bipartite graph.


We bound $\OPT{\initialf}{\targetf}$ from below as follows.
Let $\calS$ be a swapping sequence between $\initialf$ and $\targetf$.
The swapping sequence gives a perfect matching of $G_B$, as follows.
For each token of color 1,
we choose an edge $(x,y)$ of $G_B$ if the token is placed on $x \in X$ in $\initialf$ and on $y \in Y$ in $\targetf$.
The obtained set is a perfect matching of $G_B$.
A token corresponding to an edge $e$ in the matching 
needs $w(e)$ swaps,
and two tokens of color 1 are never swapped in $\calS$.
Therefore, for a minimum weight matching $M$ of $G_B$,
we have the following lower bound:
\[
\OPT{\initialf}{\targetf} \geq \displaystyle\sum_{e\in M} w(e).
\]


Now we describe our algorithm.
First we find a minimum weight perfect matching $M$ of $G_B$.
We then choose an edge $e$ in $M$.
Let $P_e = \seq{p_1,p_2, \ldots ,p_q}$ be a shortest path in $G$ corresponding to $e$.
We have the following lemma.

\begin{lemma}
Suppose that the two tokens on the endpoints of $P_e$ 
have different colors.
The two tokens can be swapped by $w(e)$ swaps such that the color of the token on each internal vertex does not change.
\end{lemma}

\begin{proof}
Without loss of generality,
we assume that $\initialf(p_1) = 2$ and $\initialf(p_q) = 1$ hold.
We first choose the minimum $i$ such that $\initialf(p_i) = 1$ holds.
We next move the token on $p_i$ to $p_1$ by $i-1$ swaps.
Note that, after the swaps above,
$p_1$ has a token of color 1 and
$p_j$ for each $j = 2,3,\ldots ,i$ has a token of color 2.
We repeat the same process to the subpath $\seq{p_i,p_{i+1}, \ldots ,p_q}$.
Finally, we obtain the desired \tpl.
Recall that there are only two colors on graphs,
and so the above ``color shift'' operation works.
See \figurename~\ref{fig:color-shift} for an example.
\begin{figure}[tb]
  \centerline{\includegraphics[width=0.5\linewidth]{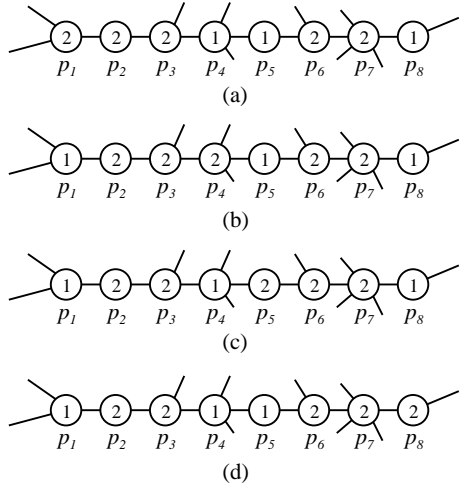}}
  \caption{
      An example of the color shift operation on
      the path $\seq{p_1,p_2,p_3,p_4,p_5,p_6,p_7,p_8}$.
      (a) The initial \tpl. (b) The \tpl\ obtained from (a) by moving the token on $p_4$ to $p_1$. (c) The \tpl\ obtained from (b) by moving the token on $p_5$ to $p_4$. (d) The \tpl\ obtained from (c) by moving the token on $p_8$ to $p_5$. The total number of swaps is 7.
  }
\label{fig:color-shift}
\end{figure}
Since each edge of $P_e$ is used by one swap,
the total number of swaps is $w(e) = q-1$.
\end{proof}

This lemma permits to move the two tokens on the two endpoints $p_1$ and $p_q$ of $P_e$ to their target positions in $w(e)$ swaps without changing the token-placement of the other vertices.
Let $\mapg$ be the \tpl\ obtained after the swaps.
We can observe that $\initialf(v)=\mapg(v)$ for every $v \in V\setminus \mset{p_1,p_q}$ and $\mapg(v) = \targetf(v)$ for $v \in \mset{p_1,p_q}$.
Then we remove $e$ from the matching $M$.
We repeat the same process until $M$ becomes empty.
Our algorithm always exchanges tokens on two vertices
using a shortest path between the vertices.
Hence, the length of the swapping sequence 
constructed by our algorithm 
is equal to the lower bound.

Now we estimate the running time of our algorithm.
The algorithm first constructs the weighted complete bipartite graph.
This can be done using the Floyd-Warshall algorithm,
which computes all-pairs shortest paths in a graph,
in \order{n^3} time.
Then, 
our algorithm
constructs a minimum weight perfect matching.
This can be done in \order{n^3} time \cite[p.252]{KF05}.
Finally,
for each edge in the matching,
along the corresponding shortest path,
our algorithm 
moves the tokens on the endpoints of the path in linear time. 
We have the following theorem.
\begin{theorem}\label{thm:generalgraphs}
\cts{2} is solvable in \order{n^3} time.
Furthermore, a swapping sequence of the minimum length
can be constructed in the same running time.
\end{theorem}

\subsection{\cts{2} on trees}

In the previous subsection,
we proved that, for general graphs, \cts{2} can be solved
in \order{n^3} time.
In this subsection, we show that \cts{2} for trees
can be solved in linear time without 
constructing a swapping sequence.

Let $T$ be an input tree, and let $\initialf$ and $\targetf$ be an initial \tpl\ and a target \tpl\ of $T$.
Let $e=(x,y)$ be an edge of $T$.
Removal of $e$ disconnects $T$ into two subtrees:
$T(x)$ with $x$ and $T(y)$ with $y$.

Now we define the value $\diff{e}$ for each edge $e$ of $T$.
Intuitively, $\diff{e}$ is the number of tokens of color 1 
which we wish to move along $e$.
More formally, we give the definition of $\diff{e}$ as follows.
Let $\ntokensinitial$ and $\ntokenstarget$ be the numbers of tokens of color 1
in $\initialf$ and $\targetf$ in $T(x)$, respectively.
Then, we define $\diff{e} = \msize{\ntokensinitial - \ntokenstarget}$.
(Note that, even if we count tokens of color 1 in $\initialf$ and $\targetf$ in $T(y)$
instead of $T(x)$, the value $\diff{e}$ takes the same value as $\initialf \simeq \targetf$.)
See \figurename~\ref{fig:diff} for an example.
\begin{figure}[tb]
  \centerline{\includegraphics[width=0.8\linewidth]{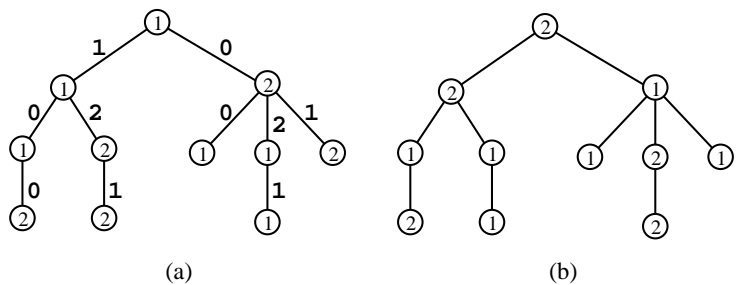}}
\caption{An example of $\diff{e}$.
  For each edge $e$, the value $\diff{e}$ is written beside $e$. \quad
  (a) The initial token-placement.
  (b) The target token-placement.}
\label{fig:diff}
\end{figure}
For each edge $e$ of $T$,
we need to move at least $\diff{e}$ tokens of color 1 
from a subtree to the other along $e$.
Therefore, 
$\OPT{\initialf}{\targetf}$ is lower bounded by the sum $D = \sum_{e \in E(T)} \diff{e}$.

To give an upper bound of $\OPT{\initialf}{\targetf}$,
we next show that there exists a swapping sequence of length $D$.
The following lemma is
the key to construct the swapping sequence.

\begin{lemma}\label{lem:edge}
If $D \ne 0$, then there exists an edge $e$ such that the swap on $e$ 
decreases $D$ by one.
\end{lemma}

\begin{proof}
We first give an orientation of edges of $T$.
For each edge $e=(x,y)$,
we orient $e$ from $x$ to $y$ if the number of tokens of color 1 in $\initialf$ in $T(x)$ is greater than the number of tokens of color 1 in $\targetf$ in $T(x)$.
Intuitively, the direction of an edge means
that we need to move one or more tokens of color 1 from $T(x)$ to $T(y)$.
If the two numbers are equal, we remove $e$ from $T$.
Let $T'$ be the obtained directed forest.
For an edge $e=(x,y)$ oriented from $x$ to $y$,
if $x$ has a token of color 1 and 
$y$ has a token of color 2, 
swapping the two tokens decreases $D$ by one.
We call such an edge a \textem{desired} edge.
We now show that there exists a desired edge in $T'$.
Observe that if no vertex $u$ with $\initialf(u) = 1$ is incident to a directed edge in $T'$, then indeed $T'$ has no edge and $D = 0$.
Let $u$ be a vertex with $\initialf(u) = 1$ that has at least one incident edge in $T'$.
If $u$ has no out-going edge, then the number of the color-1 tokens in $\initialf$ exceeds the number of the color-1 tokens in $\targetf$.
Thus we can choose an edge $(u,v)$ oriented from $u$ to $v$.
If $\initialf(v) = 2$ holds, the edge is desired.
Now we assume that $\initialf(v) = 1$ holds.
We apply the same process for $v$, then an edge $(v,w)$ oriented from $v$ to $w$
can be found.
Since trees have no cycle, by repeating the process, we always find a desired edge.
\end{proof}

From Lemma~\ref{lem:edge},
we can find a desired edge, and we swap the two tokens 
on the endpoints of the edge.
Since a swap on a desired edge decreases $D$ by one,
by repeatedly swapping on desired edges, 
we obtain the swapping sequence of length $D$.
Note that $D=0$ if and only if $\initialf(v) = \targetf(v)$ for every $v \in V(T)$.
Hence, we have $\OPT{\initialf}{\targetf} \leq D$.

Therefore $\OPT{\initialf}{\targetf} = D$ holds, 
and so we can solve \cts{2} by calculating $D$.
The value $\diff{e}$ for every edge $e$, 
and thus the value $D$,
can be calculated in a bottom-up manner
from the edges incident to leaves of $T$
in linear time in total.
We have the following theorem.

\begin{theorem}\label{thm:tree}
\cts{2} is solvable in linear time for trees.
\end{theorem}

\subsection{\cts{c} on complete graphs}

In the previous subsections, we showed that
\cts{2} can be solved in polynomial time for graphs
and linear time for trees.
In this subsection,
let us consider if \cts{c} for constant $c > 2$ can be solved
in polynomial time for restricted graphs classes.
We show that, for complete graphs,
the \cts{c} problem
is fixed-parameter tractable when $c$ is the parameter.%

%
%
As a preliminary of this subsection,
we first introduce a destination graph which represents
the target vertex of each token.
Let $G=(V,E)$ be a complete graph, and let $\initialf$ and $\targetf$ be
an initial \tpl\ and a target one.
The \textem{destination graph} $\dgraph{\initialf}{\targetf} = (\dvset, \deset)$
of two token-placement $\initialf$ and $\targetf$ is
%
%
the directed graph such that
\begin{listing}{aaa}
\item[$\bullet$] $\dvset = V$; and
\item[$\bullet$] there is an arc $(u, v)$ from $u$ to $v$ if and only if $\initialf(u) = \targetf(v)$.
\end{listing}

\subsection*{Upper bound}

We here design an algorithm that computes a swapping sequence
between $\initialf$ and $\targetf$ to give an upper bound.
Let us show an observation before describing the details of
the algorithm.
Let $\cycle$ be a cycle in $\dgraph{\initialf}{\targetf}$,
and we assume that $\cycle$ is not a self-loop.
We can move every token in $\cycle$ to its target vertex
so that the tokens on the vertices in the other cycles are unchanged, as follows.
First we choose any vertex $v$ in $\cycle$.
We swap the two tokens on $v$ and the out-neighbor $u$ of $v$ in $\cycle$.
After this swap, $\cycle$ is split into the two cycles $\cycle_1$ and $\cycle_2$:
$\cycle_1$ is a self-loop, that is the cycle with only $u$,
and $\cycle_2$ is the cycle with all the vertices in $\cycle$ except $u$.
Then, $u$ has its desired token and hence the remaining task
is to move every token on the vertices in $\cycle_2$.
We repeat to swap the two tokens on $v$ and
its out-neighbor until $v$ has its desired token.
If a cycle has 3 or more vertices,
the swap split the cycle into the two cycles:
a self loop and a cycle with one less edges.
If a cycle has 2 vertices,
the swap split the cycle into two self-loops.
Hence, the number of swaps to move all tokens in $\cycle$
to their target vertices is $|V(\cycle)| - 1$.

We now describe the details of our algorithm.
The algorithm uses a (vertex-disjoint) cycle cover of
vertices in
the destination graph.
Let $\ccover{\initialf}{\targetf}$ be a
cycle cover of $\dgraph{\initialf}{\targetf}$
and let $\cycle$ be a cycle in $\ccover{\initialf}{\targetf}$.
(We will present how to find a cycle cover later.)
The algorithm swaps tokens cycle by cycle.
If $\cycle$ is a self-loop, the algorithm does nothing.
Let us assume that $\cycle$ is not a self-loop.
The algorithm moves all tokens in $\cycle$ to
their target vertices by $\msize{V(\cycle)} -1$ swaps
while keeping the rest unchanged.
The algorithm
repeats the same process for every cycle
in $\ccover{\initialf}{\targetf}$.
The total number of swaps is
\begin{align*}
\sum_{\cycle \in \ccover{\initialf}{\targetf}} \left( \msize{V(\cycle)} -1 \right)
= \sum_{\cycle \in \ccover{\initialf}{\targetf}}\msize{V(\cycle)} - \msize{\ccover{\initialf}{\targetf}}
= n - \msize{\ccover{\initialf}{\targetf}}.
\end{align*}
Thus we have the following upper bound.
\begin{lemma}\label{lem:upper}
  Let $G$ be a complete graph with $n$ vertices, and let
  $\initialf$ and $\targetf$ be
  initial and target \tpls. If there exists a vertex
  disjoint cycle cover $\ccover{\initialf}{\targetf}$ of
  $\dgraph{\initialf}{\targetf}$, we have the following inequality:
  \[
  \OPT{\initialf}{\targetf} \leq n - \msize{\ccover{\initialf}{\targetf}}.
  \]
\end{lemma}

If we apply the above lemma to a cycle cover such that
the number of cycles in the cover is maximized,
we obtain the smallest upper bound among cycle covers
of $\dgraph{\initialf}{\targetf}$.
We call such a cycle cover \textem{optimal}.

\begin{corollary}
  Let $\ccoveropt{\initialf}{\targetf}$ be an optimal cycle cover. Then we have
  $$
    \OPT{\initialf}{\targetf} \leq n - \msize{\ccoveropt{\initialf}{\targetf}}.
  $$
\end{corollary}

Next, we show that the problem of finding an optimal cycle
cover is fixed-parameter tractable
when parameterized by the number of colors.

\subsection*{Finding an optimal cycle cover}

We show that, when $\ncolors$ is the parameter,
an optimal cycle cover of a destination graph can be found
in FPT time.
We reduce the problem into an integer linear program
such that the number of variables and the number of constraints
are both upper bounded by functions of $c$ and all values in the program are nonnegative and at most $n$.

\begin{lemma}
\label{lem:atmostc}
If $\cycle$ is a cycle in an optimal cycle cover $\ccoveropt{\initialf}{\targetf}$,
then $\cycle$ contains at most $c$ vertices.
\end{lemma}
\begin{proof}
Let $\cycle = \seq{w_1,w_2,\ldots,w_k}$.
Suppose to the contrary that $k > c$.
Then, there are two vertices with the same target color.
Without loss of generality, assume that $\targetf(w_{1}) = \targetf(w_{i}) = a$.
This implies that $\initialf(w_{k}) = \initialf(w_{i-1}) = a$.
Hence, there exist the edges $(w_{k}, w_{i})$ and $(w_{i-1}, w_{1})$
in the destination graph.
See \figurename~\ref{fig:atmost-c-vertices}.
\begin{figure}[tb]
  \centerline{\includegraphics[width=0.3\linewidth]{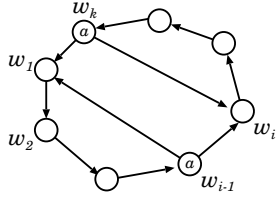}}
  \caption{%
      An illustration for Lemma~\ref{lem:atmostc}.
  }
  \label{fig:atmost-c-vertices}
\end{figure}
Therefore, the vertices of $\cycle$ can be cover by the two disjoint cycles
$\cycle_1=\seq{w_1,w_2,\ldots ,w_{i-1}}$ and $\cycle_2=\seq{w_i, w_{i+1}, \ldots ,w_k}$.
This contradicts the optimality of $\ccoveropt{\initialf}{\targetf}$.
\end{proof}

We define ``type'' of vertices in a graph.
The \textem{type} of a vertex $v$ is $(a,b)$
if $\targetf(v) = a$ and $\initialf(v) = b$ hold.
Intuitively, a vertex of type $(a,b)$ has a token of color $b$
in $\initialf$ but the vertex desires to have a token of color $a$ in $\targetf$.
The number of all possible types is $\ncolors^2$.

Now we reduce the problem of finding an optimal cycle cover
to an integer linear program as follows.
We first define a type sequence of a cycle.
For a cycle $\cycle = \seq{w_1,w_2,\ldots,w_k}$,
the \textem{type sequence} of $\cycle$ is
a sequence $\seq{tp(w_1),tp(w_2),\ldots,tp(w_k)}$,
where $tp(w_i)$ is the type of $w_i$.
We list all possible type sequences
that appear in $\dgraph{\initialf}{\targetf}$.
The length of a type sequence is
at most $\ncolors$ from Lemma~\ref{lem:atmostc}, and
the number of type sequences is at most $c!$.
For each type sequence $s$, we prepare a variable $x_s$
which represents how many times the corresponding type sequences
appear in a cycle cover.
Let $S$ be the set of type sequences.
Finally, we have the following integer linear program:
\begin{align*}
  \text{Maximize} & & & \sum_{s \in S} x_{s},\\
  \text{Subject to} & & & \sum_{\substack{s \in S \text{ with } \\ (a,b) \in s}} x_s = \#(a,b) \quad \text{ for each } (a,b), \\
  & & & x_s \geq 0 \quad \text{ for each } s \in S,
\end{align*}
where $\#(a,b)$ is the number of vertices of type $(a,b)$
in $\dgraph{\initialf}{\targetf}$.
This formula can be constructed in FPT time with parameter $c$.

The feasibility test of an ILP formula is fixed-parameter tractable when parameterized by the number of variables~\cite{FT87,K87,L83}.
For our purpose, we need to solve the optimization version.
We formally describe the problem and the theorem presented by Fellows et al.~\cite{FLMRS08} below.

\def\vector#1{\mbox{\boldmath $#1$}}

\medskip
\noindent
\textbf{Problem:} \textsc{$p$-Variable Integer Linear Programming Optimization} (\poilp) \\
\textbf{Instance:}
A matrix $\vector{A} \in \mathbb{Z}^{m\times p}$, and vectors $\vector{b} \in \mathbb{Z}^m$ and $\vector{c} \in \mathbb{Z}^p$.\\
\textbf{Objective:} 
Find a vector $\vector{x} \in \mathbb{Z}^p$ that minimizes $\vector{c}^{\top}\vector{x}$ and satisfies that $\vector{A}\vector{x} \geq \vector{b}$.\\
\textbf{Parameter:}
$p$, the number of variables
\medskip

\begin{theorem}[Fellows et al.~\cite{FLMRS08}]\label{thm:poilp}
  \poilp\ can be solved using
  \order{p^{2.5p+\mbox{\cal o}(p)} \cdot L \cdot \log{MN}}
  arithmetic operations and space polynomial in $L$,
  where $L$ is the number of bits in the input,
  $N$ is the maximum absolute values any variable can take,
  and $M$ is an upper bound on the absolute
  value of the minimum taken by the objective function.
\end{theorem}

In our ILP formulation,
(1) the number of variables is at most $c!$,
(2) the constraints are represented using \order{f(c) \log{n}} bits for some computable $f$,
(3) each variable takes value at most $n$,
and (4) the value of the objective function is at most $n$.
Therefore, by Theorem~\ref{thm:poilp}, the ILP formulation can be solved in FPT time.

\subsection*{Lower bound}

Let $G$ be a complete graph and let $\initialf$ and $\targetf$
be an initial \tpl\ and a target \tpl.
Let $\ccoveropt{\initialf}{\targetf}$ an optimal cycle cover
of the destination graph $\dgraph{\initialf}{\targetf}$.
Now we define a potential function for $\initialf$ and $\targetf$:
\begin{equation*}
  \pfunc{\initialf}{\targetf} = n - \msize{\ccoveropt{\initialf}{\targetf}}.
\end{equation*}
Note that $\pfunc{\targetf}{\targetf} = 0$ holds.

Let $\mapf$ be a token-placement of $G$, and
let $\ccover{\mapf}{\targetf}$ be a cycle cover of $\dgraph{\mapf}{\targetf}$.
Then, for $\ccover{\mapf}{\targetf}$, we define an adjacent
cycle cover as follows.
Let $\mapfp$ be a token-placement adjacent to $\mapf$,
and assume that $\mapfp$ is obtained from $\mapf$ by swapping along an edge $e=(u,v)$.
Let $(u,u')$ and $(v,v')$ be the directed edges 
leaving from $u$ and $v$, respectively, in $\ccover{\mapf}{\targetf}$.
We define $\ccoverp{\mapfp}{\targetf}$ as the cycle cover obtained
by replacing $(u,u')$ and $(v,v')$ with $(v,u')$ and $(u,v')$.
Note that $\ccoverp{\mapfp}{\targetf}$ is a cycle cover of $\dgraph{\mapf}{\targetf}$.
Then, we say that
$\ccoverp{\mapfp}{\targetf}$ is \textem{adjacent} to
$\ccover{\mapf}{\targetf}$ with respect to $e$.
Note that any cycle cover of $\dgraph{\mapfp}{\targetf}$ is adjacent to some cycle cover of $\dgraph{\mapf}{\targetf}$
with respect to $e$.

\begin{lemma}\label{lem:minusone}
For any adjacent two \tpls\ $\mapf$ and $\mapfp$,
$\pfunc{\mapfp}{\targetf} \geq \pfunc{\mapf}{\targetf} - 1$ holds.
\end{lemma}

\begin{proof}
  Suppose $\mapfp$ is obtained from $\mapf$
  by swapping along an edge $e=(u,v)$.
  Let $\ccover{\mapf}{\targetf}$ be a cycle cover of
  $\mapf$, and let $\ccoverp{\mapfp}{\targetf}$ be its adjacent
  cycle cover with respect to $e$.
  Now, we investigate how the number of cycles in $\ccover{\mapf}{\targetf}$ changes by swapping along $e$.

  \begin{mycase}{1}{$u$ and $v$ are in the same cycle.
      See \figurename~\ref{fig:split-merge-cycle}(a).}
    Swapping along $e$ splits the cycle including $u$ and $v$
    into two cycles.
    Hence, the number of cycles increases by one.
  \end{mycase}
  
  \begin{mycase}{2}{$u$ and $v$ are in different cycles.
    See \figurename~\ref{fig:split-merge-cycle}(b).}
    Swapping along $e$ combines the two cycle
    including $u$ and $v$ into one cycle.
    Hence, the number of cycles decreases by one.
  \end{mycase}
  \vspace{2mm}

  \begin{figure}[tb]
    \centerline{\includegraphics[width=0.9\linewidth]{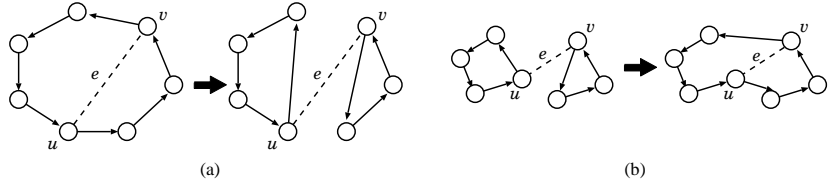}}
    \caption{%
        (a) An illustration for Case~1. The cycle in the left side is split into the two cycles in the right side. (b) An illustration for Case~2. The two cycles in the left side are combined into the cycle in the right side.
    }
    \label{fig:split-merge-cycle}
  \end{figure}
  
  Hence, we have $\msize{\ccover{\mapfp}{\targetf}} \leq \msize{\ccover{\mapf}{\targetf}} + 1$.
  Recall that any cycle cover of $\dgraph{\mapfp}{\targetf}$ is adjacent to
  a cycle cover of $\dgraph{\mapf}{\targetf}$.
  Thus we also have $\msize{\ccoveropt{\mapfp}{\targetf}} \leq \msize{\ccoveropt{\mapf}{\targetf}} + 1$.
  Therefore, $\pfunc{\mapfp}{\targetf} \geq \pfunc{\mapf}{\targetf} - 1$ holds.
\end{proof}

From Lemma~\ref{lem:minusone},
we have the following lower bound.
\begin{lemma}\label{lem:lower}
Let $G$ be a complete graph and let $\initialf$ and $\targetf$
be an initial \tpl\ and a target one.
Then we have
$$
\OPT{\initialf}{\targetf} \geq n - \msize{\ccoveropt{\initialf}{\targetf}}.
$$
\end{lemma}

Therefore, we have the following theorem.
\begin{theorem}\label{thm:clique}
  Given a complete graph, an initial \tpl\ $\initialf$
  and a target one $\targetf$,
  \cts{c} is fixed-parameter tractable when $c$ is the parameter.
\end{theorem}

\bibliographystyle{abbrv}
\bibliography{final}

\end{document}